\documentclass{sig-alternate-05-2015}

\usepackage{amsmath,amsfonts,amssymb}
\usepackage{url}
\usepackage{graphicx}
\usepackage{algorithm}
\usepackage{algorithmic}
\usepackage{float}
\usepackage{multirow}
\usepackage{graphicx}
\usepackage{hyperref}
\usepackage{pgfplots}
\usepackage{tikz}
\usepackage{pgfplots}
\usepackage{pgf}
\usepackage{placeins}
\usepackage{subcaption}
\usepackage{comment}
\usepackage{amsmath}
\usepackage{url}
\usepackage{placeins}
\usepackage{latexsym}
\usepackage{pifont}
\usepackage{authblk}

\usetikzlibrary{intersections}
\usetikzlibrary{arrows,calc,fit,patterns,plotmarks,shapes.geometric,
shapes.misc,shapes.symbols,   shapes.arrows,   shapes.callouts,
shapes.multipart,   shapes.gates.logic.US,   shapes.gates.logic.IEC,   
er,   automata,   backgrounds,   chains,   topaths,   trees,   
petri,   mindmap,   matrix,   calendar,   folding, fadings,   
through,   positioning,   scopes,   decorations.fractals,   
decorations.shapes,   decorations.text,   decorations.pathmorphing,   
decorations.pathreplacing,   decorations.footprints,   
decorations.markings, shadows,circuits}

\tikzstyle{decision}=[diamond,draw]
\tikzstyle{line}=[draw]
\tikzstyle{elli}=[draw,ellipse]
\tikzstyle{arrow} = [thick]

\newcommand{\nn}{\nonumber}
\newcommand{\R}{\mathbf{R}}

\newcommand{\ra}{\rightarrow}
\newcommand{\thetaa}{\theta_{A}}
\newcommand{\thetah}{\theta_{H}}
\newcommand{\thetahn}{\theta^{\min}_{H}}
\newcommand{\thetahx}{\theta^{\max}_{H}}
\newcommand{\thetahd}{\theta^d_{H}}
\newtheorem{theorem}{Theorem}

\newtheorem{assumption}{Assumption}

\newtheorem{example}{Example}

\definecolor{color1}{rgb}{0.301961,0.686275,0.290196}

\begin{document}
\title{Performance Tuning of Hadoop MapReduce: A Noisy Gradient Approach}

\author{Sandeep Kumar}
\author{Sindhu Padakandla}
\author{Chandrashekar L}
\author{Priyank Parihar}
\author{K. Gopinath}
\author{Shalabh Bhatnagar}

\affil{Department of Computer Science and Automation\\
 Indian Institute of Science \\
\email{\small\{sandeep007734, sindhupr, chandrurec5, iisc.csa.priyank.parihar\}@gmail.com, \{gopi, shalabh\}@csa.iisc.ernet.in }
}


\maketitle
\begin{abstract}
Hadoop MapReduce is a framework for distributed storage and processing of large
datasets that is quite popular in big data analytics. 
It has various configuration parameters 
(knobs) which play an important role in deciding the performance i.e., the 
execution time of a given big data processing job. 
Default values of these parameters do not always 
result in good performance and hence it is important to tune them. However, 
there is inherent difficulty in tuning the parameters due to two important reasons - 
firstly, the parameter search space is large and secondly, there are cross-parameter interactions.
Hence, there is a need for a dimensionality-free method which can 
automatically tune the configuration parameters by taking into account the cross-parameter
dependencies. In this paper, we propose a novel Hadoop parameter tuning methodology, 
based on a noisy gradient algorithm known as the simultaneous perturbation stochastic approximation (SPSA).
The SPSA algorithm tunes the parameters by directly observing the 
performance of the Hadoop MapReduce system. The approach followed is independent of parameter 
dimensions and requires only $2$ observations per iteration while tuning. 
We demonstrate the effectiveness of our methodology in achieving good performance
on popular
Hadoop benchmarks namely \emph{Grep}, \emph{Bigram}, \emph{Inverted Index}, \emph{Word Co-occurrence}
and \emph{Terasort}.
Our method, when tested on a 25 node Hadoop cluster shows 66\% decrease 
in execution time of Hadoop jobs on an average, when compared to the default configuration. 
Further, we also observe a reduction of 45\% in execution times, when compared to prior methods.
\end{abstract}	

\section{Introduction}
\label{intro}
We are in the era of big data and huge volumes of data are generated in various 
domains like social media, financial markets, transportation and health care. 
Faster analysis of such big unstructured data is a key requirement for achieving 
success in these domains. Popular instances of such cases include distributed 
pattern-based searching, distributed sorting, web link-graph reversal, singular 
value decomposition, web access log stats, inverted index construction and 
document clustering. Extracting  hidden patterns, unknown correlations and other 
useful information is critical for making better decisions. 
Many industrial organisations like Yahoo!, Facebook, Amazon etc. need to handle
and process large volumes of data and their product success hinges on this ability.
Thus, there is a 
need for parallel and distributed processing/programming methodologies that can 
handle big data using resources built out of commodity hardware. Currently available 
parallel processing systems are database systems \cite{pavlo} like Teradata, Aster Data, Vertica etc.,
which are quite robust and are high-performance computing platforms. 
However, there is a need for a parallel processing system which
can handle large volumes of data using low-end servers and which is easy-to-use.
\textbf{MapReduce}\cite{ghemawat} is one-such programming model.

MapReduce computation over input data goes through two phases namely \emph{map} 
and \emph{reduce}. At the start of the map phase, the job submitted by the client is 
split into multiple map-reduce tasks that are to be executed by  
various worker nodes. The map phase then 
creates the key-value pairs from the input dataset according to the user defined 
map. The reduce phase makes use of the key-value pairs and aggregates according 
to user specified function to produce the output.

Apache \textbf{Hadoop}\cite{hadoop} is an open-source implementation of MapReduce written in 
Java for distributed storage and processing of very large data sets on clusters 
built using commodity hardware.
The Hadoop framework gives various parameter (knobs) that need to be tuned 
according to the program, input data and hardware resources. It is important to 
tune these parameters to obtain best performance for a given MapReduce job.
The problem of Hadoop performance being limited by the parameter configuration
was recognized in \cite{kambatla}.
Unlike SQL, MapReduce jobs cannot be modeled using a small and finite 
space of relational operators \cite{pavlo}. Thus, 
it is not straight forward to quantify the effect of these various parameters on 
the performance and hence it is difficult to compute the best parameter 
configuration apriori. In addition, difficulty in tuning these parameters also 
arises due to two other important reasons. Firstly, due to the presence of a 
large number of parameters (about $200$, encompassing a variety of functionalities) 
the search space is large and 
complex.
Secondly, there is a pronounced effect of cross-parameter interactions, i.e., 
the knobs are not independent of each other. For instance, increasing the 
parameter corresponding to map-buffer size will decrease the I/O cost, however, 
the overall job performance may degrade because sorting cost may increase (in 
quick sort, sorting cost is proportional to the size of data). The complex 
search space along with the cross-parameter interaction does not make Hadoop 
amenable to manual tuning.

The necessity for tuning of Hadoop parameters was first emphasized in 
\cite{kambatla}, which proposed a method to determine the optimum
configuration given a set of computing resources.
Recent efforts in the direction of automatic tuning of the Hadoop parameters include 
Starfish\cite{starfish}, AROMA\cite{aroma},
MROnline\cite{mronline}, 
PPABS \cite{ppabs} and JellyFish \cite{jellyfish}. 
We observe that collecting statistical data to create virtual profiles and 
estimating execution time using mathematical  model (as in 
\cite{kambatla,starfish,jellyfish,aroma,mronline,ppabs}) requires significant level of expertise which 
might not be available always. In addition, since Hadoop MapReduce is evolving 
continuously with a number of interacting parts, the mathematical model also has 
to be updated and in the worst case well-defined mathematical model might not be 
available for some of its parts due to which a model-based approach might fail. 
Further, given the presence of cross-parameter interaction it is a good idea to 
retain as many parameters as possible (as opposed to reducing the parameters 
\cite{ppabs}) in the tuning phase.

In this paper, we present a novel tuning methodology based on a noisy gradient 
method known as the simultaneous perturbation stochastic approximation (SPSA) 
algorithm \cite{spall}. The SPSA algorithm is a black-box stochastic optimization 
technique which has been applied to tune parameters in a variety of complex 
systems. 
An important feature of the SPSA is that it 
utilizes observations from the real system as feedback to tune the 
parameters. Also, the SPSA algorithm is dimensionality free, i.e., it needs only 
$2$ or fewer observations per iteration irrespective of the number of parameters 
involved. In this paper, we adapt the SPSA algorithm to tune the 
parameters used by Hadoop to allocate resources for program execution.

\subsection{Our Contribution}
Our aim is to introduce the practitioners to a new method that is different 
in flavour from the prior methods, simple to implement and effective at the same 
time. The highlights of our SPSA based approach are as follows:
\begin{itemize}
\item \textbf{Mathematical model}: The methodology we propose utilizes the 
observations from the Hadoop system and does not need a mathematical model. This is desirable since 
mathematical models developed for older versions might not carry over to newer 
versions of Hadoop.
\item \textbf{Dimension free nature}: SPSA is designed to handle complex 
search spaces. Thus, unlike \cite{Babu10towardsautomatic} reducing the search 
space is not a requirement.
\item \textbf{Parametric dependencies}: Unlike a host of black-box optimization 
methods that depend on clever heuristics, our SPSA based method computes 
the \emph{gradient} and hence takes 
into account the cross parameter interactions in the underlying problem.
\item \textbf{Performance}: Using the SPSA algorithm we tune $11$ parameters 
simultaneously. 
Our method provides a 66\% decrease 
in execution time of Hadoop jobs on an average, when compared to the default configuration. 
Further, we also observe a reduction of 45\% in execution times, 
when compared to prior \cite{starfish} methods.
%
\end{itemize}

\subsection{Organisation of the Paper}
In the next section, we describe the Hadoop architecture, its data flow analysis and 
point out the importance and role of some of the configuration parameters.
Following it, in Section~\ref{relatedwork} we discuss the related work 
and contrast it with our approach. We provide a detailed description of our 
SPSA-based approach in Section \ref{autotune}.
In Section \ref{setup} we discuss the specific details in implementing the SPSA 
algorithm to tune the Hadoop parameters.
We describe the experimental setup and present the results in Section \ref{exp}. 
Section \ref{concl} concludes the paper and suggests future enhancements.

\section{Hadoop}
\label{hadoop}
Hadoop is an open source implementation of the MapReduce\cite{ghemawat}, 
which has gained a huge amount of 
popularity in recent years as it can be used over commodity hardware. Hadoop 
 has two main components namely MapReduce and Hadoop Distributed File 
System(HDFS). The HDFS is used for storing data and MapReduce is used for 
performing computations over the data. 
We first discuss the HDFS and then MapReduce. Following this, we describe
the data flow analysis in Hadoop with an aim
to illustrate the importance of the various parameters.

\subsection{Hadoop Distributed File System}
Hadoop uses HDFS to store input and 
output data for the MapReduce applications. HDFS provides interfaces for 
applications to move themselves closer \cite{url_6} to where the data is located 
because data movement will be costly as compared to movement of small MapReduce 
code. It is fault tolerant and is optimized for storing large data sets. 

A HDFS cluster (see \cite{hadoop}) consists of a single NameNode, a 
master server, and multiple slave DataNodes. 
The DataNodes, usually one per node, store the actual data 
used for computation. These manage the storage attached to the nodes 
that they run on. 
Internally, a file is split into one or more data 
blocks (block size is controlled by \emph{dfs.block.size}) and these blocks 
are stored in a set of DataNodes. They are responsible for serving read and 
write requests from the file system's clients.
NameNode manages the file system namespace and regulates access to 
files by clients. It has the following functions:
\begin{itemize}
 \item Store HDFS metadata and execute file systems operations on HDFS
 \item Mapping data blocks to DataNodes
 \item Periodically monitor the performance of DataNodes
\end{itemize}

\subsection{MapReduce}
A client application submits a MapReduce job. It is then split into various map 
and reduce tasks that are to be executed in the various cluster nodes. 
In MapReduce version 1 (v1), the JobTracker, usually running on a dedicated node, 
is responsible for execution and monitoring of jobs in the cluster.
It schedules map and reduce tasks to be run on the nodes in the cluster, 
which are monitored by a corresponding TaskTracker running on that particular node.
Each TaskTracker sends the progress of the corresponding map or reduce task to JobTracker
at regular intervals.
Hadoop MapReduce version 2 (v2, also known as Yet Another Resource Negotiator (YARN)\cite{yarn}) 
has a different architecture. 
It has a ResourceManager and NodeManager instead of JobTracker and TaskTracker. 
The tasks of resource and job management are distributed among resource manager and 
application master (a process spawned for every job) respectively.
The job submitted by a client application (for e.g., Terasort, WordCount benchmark applications) 
is associated with a NodeManager, which starts an ``Application Master'' in a container 
(a container is a Unix process, which runs on a node). The container architecture
utilizes cluster resources better, since YARN
manages a pool of resources that can be allocated based on need. This is unlike 
MapReduce v1 where each TaskTracker is configured with an inflexible map/reduce slot.
A map slot can only be used to run a map task and same with reduce slots.

\subsection{MapReduce Data Flow Analysis}
Map and Reduce are the two main phases of job processing (see Fig. \ref{hadooponearch}). 
The function of these phases is 
illustrated with the following simple example: 
\begin{example}
The objective is to count the number of times each word appears in a file whose 
content is given by,
\begin{center}
``This is an apple. That is an apple''.
\end{center}
The output of the Map operation is then given by, 
\begin{center}
	$<This,1> <is,1> <an,1><apple,1>$\\$<That,1><is,1><an,1><apple,1>$, 
\end{center}
following which the Reduce operation outputs
\begin{center}
    $<This, 1><That, 1> <is, 2> <an, 2> <apple, 2>$. 
\end{center}
Thus we obtain the count for each of the words. 
\end{example}
Map and Reduce phases can perform complex computations in contrast to the above example.
The efficiency of these computations and phases is controlled by various system parameters.
We describe the parameters our algorithm tunes (see Section \ref{exp}) and show how 
they influence the map and reduce phases of computation.

\begin{figure*}
\centering
\includegraphics[scale=0.22]{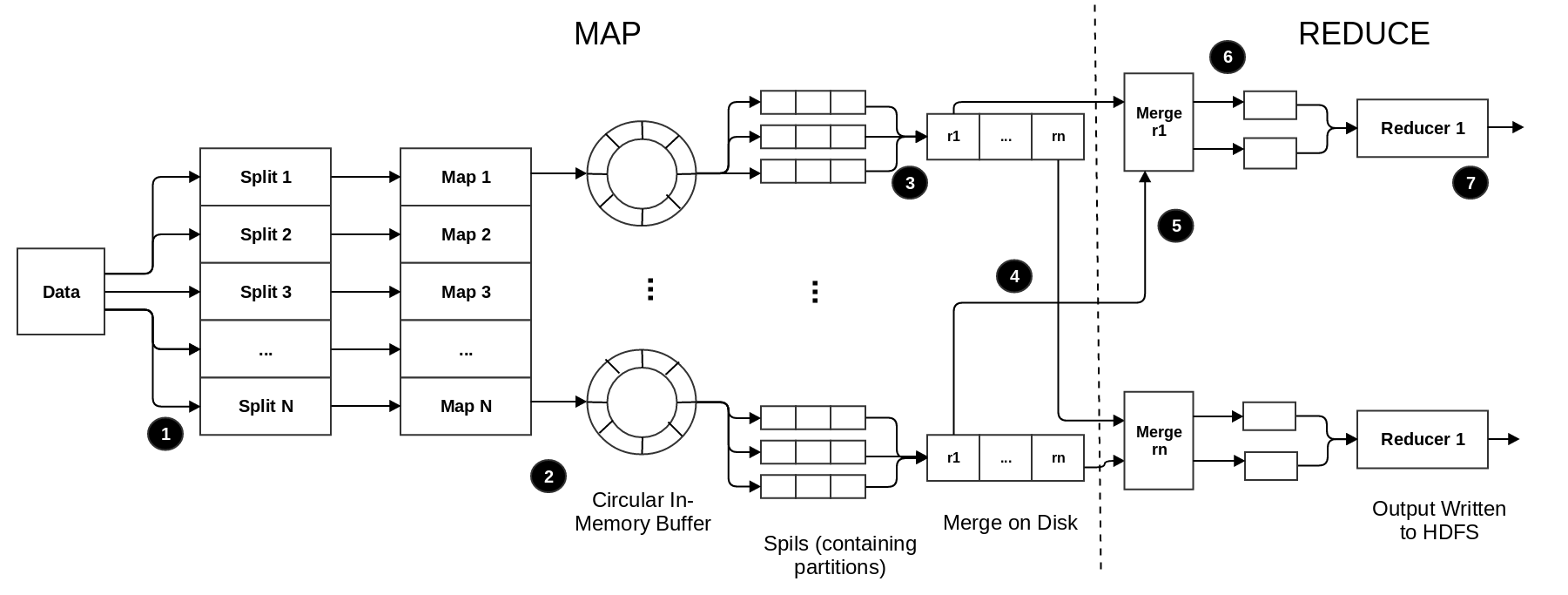}
  \caption{MapReduce working}
  \label{hadooponearch}
\end{figure*}

\subsubsection{Map Phase}
Input data is split according to the input format (text, zip etc.) and split
size is controlled by the parameter \emph{dfs.block.size}.
For each split, a corresponding mapper task is created.

The \emph{Map} function retrieves data (records) from the input split 
with the help of the \emph{record reader}. The record reader provides 
the key-value pair (record) to the mapper according to the input format. 
The Map outputs $<key, value>$ and the meta-data (corresponding partition number of the 
key) according to the logic written in the map function. This output is 
written to a circular buffer, whose size is controlled by parameter 
\emph{mapreduce.task.io.sort.mb}. 
When the data in the buffer reaches a 
threshold, defined by \emph{mapreduce.map.sort.spill.percent}, 
data is \emph{spilled} to the local disk of a mapper node. 
The Map outputs will continue to be 
written to this buffer while the spill takes place. If any time the buffer 
becomes full, the Map task is blocked till spill finishes.

\emph{Sorting} (default \emph{Quick Sort}) and \emph{combine} operations are 
performed in the memory buffer prior to the data spill onto the disk. So 
increasing the buffer size of mapper decreases I/O cost but sorting cost will 
increase. \emph{Combine} executes on a subset of $<key,value>$ pairs. 
Combine is used for reducing data written to the disk. 

The \emph{merge} phase starts once the Map and the Spill phases 
complete. In this phase, all the spilled files from a mapper are merged together 
to form a single output file. Number of streams to be merged is controlled by 
the parameter \emph{mapreduce.task.io.sort.factor} (a higher value means more 
number of open file handles and a lower value implies  multiple merging rounds 
in the merge phase). After merging, there could be multiple records with same 
key in the merged file, so combiner could be used again. 

\subsubsection{Reduce Phase}
\emph{Reducers} are executed in parallel to the mappers if the 
fraction of map task completed is more than the value of 
\emph{mapreduce.job.reduce.slowstart.completedmaps} parameter, ot\\herwise  
reducers execute after mappers. The number of reducers for a work is controlled 
by \emph{mapreduce.job.reducers}. A Reducer fetches its input partition from 
various mappers via HTTP or HTTPS. The total amount of memory allocated to a 
reducer is set by \emph{mapreduce.reduce.memory.totalbytes} and the fraction 
of memory allocated for storing data fetched from mappers is set by 
\emph{mapreduce.reduce.shuffle.input.buffer.}\\\emph{percent}.
In order to create a single sorted data based on the key, the merge phase is
executed in order to collate the keys obtained from different partitions.
The number of map outputs needed to start this merge process is determined by 
\emph{mapreduce.reduce.merge.inmem.threshold}. The threshold for spilling the merged
map outputs to disk is controlled by
\emph{mapreduce.reduce.shuffle.merge.percent}.
Subsequent to all merge operations, reduce code is executed and its output is saved to the HDFS 
(as shown in Figure \ref{hadooponearch}). 

\subsubsection{Cross-Parameter Interaction}
The parameters \emph{io.sort.mb},\emph{reduce.input.buffer.percent} and 
\emph{shuffle.input.buffer.percent} control the number of spills written to disk.
Increasing the memory allocated will reduce the number of spill records in both Map and Reduce phases.
When \emph{io.sort.mb} is high, the spill percentage of Map (controlled by \emph{sort.spill.percent})
should be set to a high value. In the Reduce phase, the map outputs are merged and spilled to disk
when either the \emph{merge.inmem.threshold} or \emph{shuffle.merg\\e.percent} is reached.
Similarly, the \emph{task.io.sort.factor} determines the 
minimum number of streams to be merged at once, during sorting. 
So, on the reducer side, if there are say 40 mapper outputs and this value is set to 10, 
then there will be 5 rounds of merging (on an average 10 files for merge round).

The above examples indicate that changing one function in the reduce/map phase, affects 
other characteristics of the system. This implies that Hadoop system parameters cannot be tuned in isolation.
Each parameter has to be tuned by taking note of the values of related parameters.
Our SPSA-based method takes into account such relations between the parameters and 
appropriately tunes them to achieve enhanced performance.
In the next section, we discuss the existing work in the literature
which suggest techniques to enhance the performance of Hadoop.

\section{Related work}
\label{relatedwork}
Some early works \cite{Ganapathi:EECS-2009-181,cmureport} have focussed on analysing the MapReduce performance
and not addressed the problem of parameter tuning.
The authors in \cite{Ganapathi:EECS-2009-181} develop models for predicting performance of 
Hive queries and ETL (Extract Transform Load) kind of MapReduce jobs. 
This work uses KCCA (Kernel Canonical Correlation Analysis) and nearest neighbor 
for modeling and prediction respectively. KCCA provides dimensionality reduction 
and preserves the neighborhood relationship even after projecting onto a lower 
dimensional space. It uses a number of training sets to build a single model for 
multiple performance metrics. 
Hadoop job log for a period of six months is used for training. Input data 
characteristic (like byte read locally, byte read from HDFS and byte input to a map stage), 
configuration parameters, job count (number and configuration of map and reduce to be executed 
by a given Hadoop job), query operator count etc. are used as features for comparison and 
prediction about new job.

MapReduce logs of a M45 supercomputing cluster (released by Yahoo!) are analysed in \cite{cmureport}.
This analysis characterizes job patterns, completion times, job failures and resource utilization patterns
based on the logs. Jobs are characterized into map-only, reduce-only, reduce-mostly etc.
Based on this categorization, \cite{cmureport} suggests improvements in Hadoop MapReduce
which can mitigate performance bottlenecks and reduce job failures.

Attempts toward building an optimizer for hadoop performance started with Starfish\cite{starfish}.
In Starfish \cite{starfish,journals/pvldb/HerodotouB11}, a \emph{Profiler} collects detailed 
statistical information (like data flow and cost statistics) from unmodified Mapreduce program 
during full or partial execution. Then, a \emph{What-if} engine estimates the 
cost of a new job without executing it on real system using mathematical models, simulation, 
data flow and cost statitics. The Cost-based optimizer (CBO) uses the what-if 
engine and recursive random search (RSS) for tuning the parameters for a new Mapreduce job. 

Works following Starfish are \cite{aroma,ppabs}. These methods collect information about
the jobs executed on hadoop, a process known as \emph{profiling}.
Job ``signatures'', i.e., the resource utilization patterns of the jobs are used for profiling.
In the offline phase, using a training set, 
the jobs are clustered (using variants of k-means) according to their respective signatures. 
In the online phase \cite{aroma} trains a SVM which makes accurate and fast prediction of a job's performance
for various configuration parameters and input data sizes. For any new job, its signature is matched with
the profiles of one of the clusters, after which that cluster's optimal parameter
configuration is used. In \cite{ppabs}, the optimal parameter configuration 
for every cluster is obtained through simulated annealing, albeit 
for a reduced parameter search space.

An online MapReduce performance tuner (MROnline) is developed in \cite{mronline}.
It is desgined and implemented on YARN \cite{yarn} (described in Section \ref{hadoop}). 
MROnline consists of a centralized master component which is the online
tuner. It is a daemon process that runs on the same machine as the
resource manager of YARN or on a dedicated machine. Online
tuner controls slave components that run
within the node managers on the slave nodes of the YARN cluster.
It consists of three components: a \emph{monitor}, a
\emph{tuner} and a \emph{dynamic configurator}. The monitor works together
with the per-node slave monitors to periodically monitor application statistics.
These statistics are sent to the centralized monitor. The centralized monitor then 
aggregates, analyzes and passes the information
to the tuner.
The tuner implements hill climbing algorithm to tune parameter values. The tuned parameter values
are distributed to the slave configurators by the dynamic configurator.
The slave configurators activate the new parameter values for the 
tasks that are running on their associated nodes.

Industry and MapReduce vendors also provide guides \cite{microsoft,cloudera} on parameter tuning
which help in finding suitable values for the client applications.
However, these guides are heuristic and the end-users are still faced with the challenge
of manually trying out multiple parameter configurations.

\subsection{Motivation for Our Approach}
The contrast between prior approaches to parameter tuning in Hadoop and our 
approach is shown in Figure \ref{approaches}.
\begin{figure}[h!]
\centering
\begin{tikzpicture}[domain=0:7.7,scale=0.7,font=\small,axis/.style={very thick, ->, >=stealth'}]
\node [right] at (0.5,0.25){\textbf{Prior Art}};
\node [left] at (0,-0.5){Collect Data};
\draw [->](0,-0.5)--(0.5,-0.5);
\node [right] at (0.5,-0.5){Model};
\draw [->](1.95,-0.5)--(2.25,-0.5);
\node [right] at (2.25,-0.5){Simulate};
\draw [->](4.1,-0.5)--(4.5,-0.5);
\node [right] at (4.5,-0.5){Optimize with};
\node [right] at (3.8,-1){reduced parameters};
\node [right] at (1,-1.4) {\textbf{vs.}};
\node [right] at (0,-2.3){\textbf{Our approach}};
\node [left] at (1,-2.9){Observe System};
\draw [->](1,-2.9)--(2.5,-2.9);
\node [right] at (2.5,-2.9){Optimize via Feedback};
\end{tikzpicture}
\caption{Prior art vs. our approach to parameter tuning in Hadoop.}
\label{approaches}
\end{figure}
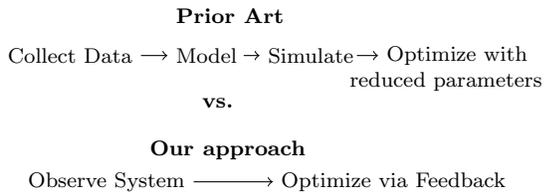 
In \cite{journals/pvldb/HerodotouB11}, the optimization is based on the 
\emph{what-if} engine which uses a mix of simulation and model-based estimation. 
Here, the cost model $F$ is high-dimensional, nonlinear, non-convex and 
multimodal. In \cite{ppabs}, authors make use of available 
knowledge from literature in order to reduce the parameter space and they make 
use of simulated annealing to find the right parameter setting in the reduced 
space. We observe that 
\begin{enumerate}
\item Collecting statistics and building an accurate model requires certain 
level of expertise. Also, mathematical models developed for an older version may 
fail for the newer versions since Hadoop is evolving continuously. In the worst 
case, mathematical models might not be well defined for some components of 
Hadoop. 
\item The effect of cross-parameter interactions are significant and hence it 
might be a good idea to have the search space as big as possible.
\end{enumerate}
With the above two points in mind, we suggest a more direct approach (see 
Figure \ref{approaches}), i.e., we suggest a method that directly utilizes the 
data from the real system and tunes the parameters via \emph{feedback}. Thus, we 
are motivated to adapt SPSA algorithm 
to tune the parameters. We believe that the SPSA based scheme is of interest to 
practitioners because it does not require any model building and it uses only the 
\emph{gradient} estimate at each step. Through the gradient estimate, it takes the cross 
parameter interaction into account. Further, the SPSA algorithm is not limited 
by the parameter dimension and requires only $2$ measurements per iteration.

\section{Automatic Parameter Tuning}
\label{autotune}
The performance of various complex systems such as  traffic control \cite{tsc}, 
unmanned aerial vehicle (UAV) control \cite{uav}, remote sensing \cite{remotesensing}, 
communication in satellites \cite{satellite} and airlines \cite{airlines} 
depends on a set of tunable parameters (denoted by $\theta$). 
Parameter tuning in such cases is difficult because of bottlenecks 
namely the \emph{black-box} nature of the problem and the 
\emph{curse-of-dimensionality} i.e., the complexity of the search space. In this section, we 
discuss the general theme behind the methods that tackle these bottlenecks and 
their relevance to the problem of tuning the Hadoop parameters.
\subsection{Bottlenecks in Parameter Tuning}
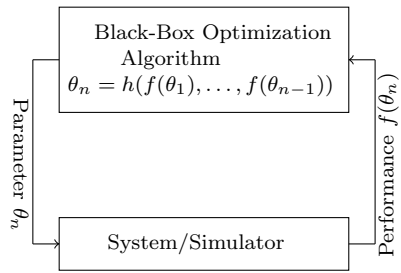
\begin{figure}[h!]
\centering
\begin{tikzpicture}[domain=0:7.7,scale=0.7,font=\small,axis/.style={very thick, ->, >=stealth'}]
\draw (0,0) rectangle (5.5,2);
\draw (0,-3) rectangle (5.5,-2);
\draw[-](0,1)--(-0.5,1);
\draw[-](-0.5,1)--(-0.5,-2.5);
\draw[->](-0.5,-2.5)--(0,-2.5);
\draw[->](6,1)--(5.5,1);
\draw[-](6,1)--(6,-2.5);
\draw[-](6,-2.5)--(5.5,-2.5);
\node [right] at (0.75,-2.5){System/Simulator};
\node [rotate=270] at (-0.75,-1) {Parameter $\theta_n$};
\node [rotate=90] at (6.25,-1) {Performance $f(\theta_n)$};
\node [right] at (1,1) {Algorithm};
\node [right] at (0.5,1.5) {Black-Box Optimization};
\node [right] at (-0.1,0.5) {\hspace{2pt}$\theta_n=h(f(\theta_1),\ldots,f(\theta_{n-1}))$\hspace{2pt}};
\end{tikzpicture}
\caption{The black-box optimization algorithm makes use of the \emph{feedback} received from the system/simulator to tune the parameters. Here $n=1,2,\ldots$ denotes the trial number, $\theta_n$ is the parameter setting at the $n^{th}$ trial and $f(\cdot)$ is the performance measure. The map $h$ makes use of past observations to compute the current parameter setting.}
\label{bbopti}
\end{figure}
In many complex systems, the exact nature of the dependence of the performance 
on the parameters is not known explicitly i.e., the performance cannot be 
expressed as an analytical function of the parameters. As a result, the 
parameter setting that offers the best performance cannot be computed apriori. 
However, the performance of the system can be observed for any given parameter 
setting either from the system or a simulator of the system. In such a scenario, 
one can resort to \emph{black-box}/\emph{simulation-based} optimization methods 
that tune the parameters based on the output observed from the system/simulator 
without knowing its internal functioning. Figure \ref{bbopti} is a schematic to 
illustrate the black-box optimization procedure. Here, the black-box 
optimization scheme sets the current value of the parameter based on the past 
observations. The way in which past observation is used to compute the current 
parameter setting varies across methods.

An important issue in the context of black-box optimization is the number of 
observations and the cost of obtaining an observation from the system/simulator. 
The term \emph{curse-of-dimensionality} denotes the exponential increase in the 
size of the search space as the number of dimensions increases. 
In addition, in many applications, the parameter $\theta$ belongs to a subset 
$X$ of $\R^n$ (for some positive integer $n>0$). Since it is computationally 
expensive to search such a large and complex parameter space, it is important for 
black-box optimization methods to make as fewer observations as possible.

Hadoop MapReduce  exhibits the above described black box kind of behavior because it is not well 
structured like SQL.
In addition, cross-parameter interactions also affect the performance, and hence 
it is not possible to treat the parameters independent of each other. Besides, the 
problem is also afflicted by the curse-of-dimesionality.
\subsection{Noisy Gradient based optimization}
In order to take the cross-parameter interactions into account, one has to make 
use of the sensitivity of the performance measure with respect to each of the 
parameters at a given parameter setting. This sensitivity is formally known as 
the \emph{gradient} of the performance measure at a given setting. It is 
important to note that it takes only $O(n)$ observations to compute the gradient 
of a function at a given point. However, even $O(n)$ computations are not 
desirable if each observation is itself costly.

Consider the noisy gradient scheme given in \eqref{ngrad} 
below.
\begin{align}\label{ngrad}
\theta_{n+1}&=\theta_n-\alpha_n\big(\nabla f_n +M_n\big),
\end{align}
where $n=1,2\ldots$ denotes the iteration number, $\nabla f_n\in \R^n$ is the 
gradient of function $f$, $M_n\in \R^n$ is a zero-mean noise sequence and $\alpha_n$ is the 
step-size. 
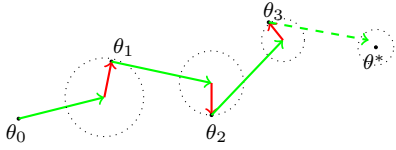
\begin{figure}[h!]
\centering
\begin{tikzpicture}[domain=0:7.7,scale=0.95,font=\small,axis/.style={very thick, ->, >=stealth'}]
\draw [fill=black](5.3,1.2) circle (0.02);
\draw [dotted](5.3,1.2) circle (0.25);
\node[right] at (5,1) {$\theta^*$};
\draw [fill=black](0.3,0.2) circle (0.02);
\node[right] at (0,0) {$\theta_0$};
\draw[green,thick,->](0.3,0.2)--(1.5,0.5);
\draw [fill=black](1.6,1) circle (0.02);
\draw [dotted](1.5,0.5) circle (0.55);
\draw[red,thick,->](1.5,0.5)--(1.6,1);
\node[right] at (1.5,1.2) {$\theta_1$};
\draw[green,thick,->](1.6,1)--(3,0.7);
\draw [dotted](3,0.7) circle (0.45);
\draw[red,thick,->](3,0.7)--(3,0.25);
\draw [fill=black](3,0.25) circle (0.02);
\node[right] at (2.8,0.0) {$\theta_2$};
\draw[green,thick,->](3,0.25)--(4,1.3);
\draw [dotted](4,1.3) circle (0.3);
\draw[red,thick,->](4,1.3)--(3.8,1.55);
\draw [fill=black](3.8,1.55) circle (0.02);
\node[right] at (3.6,1.7) {$\theta_3$};
\draw[green,thick,dashed,->](3.8,1.55)--(5.2,1.3);
\end{tikzpicture}
\caption{Noisy Gradient scheme. Notice that the noise can be filtered by an appropriate choice of diminishing step sizes.}
\label{noisegrad}
\end{figure}
Fig. \ref{noisegrad} presents an intuitive picture of how a noisy gradient algorithm works. 
Here, the algorithm starts at $\theta_0$ and needs to move to $\theta^*$ which is the 
desired solution. The green lines denote the true gradient step (i.e., $\alpha_n \nabla f_n$) 
and the dotted circles show the region of uncertainty due to the noise term $\alpha_n M_n$. 
The red line denotes the fact that the true gradient is disturbed and the iterates are 
\emph{pushed} to a different point within the region of uncertainty. The idea 
here is to use diminishing step-sizes to filter the noise and eventually move 
towards $\theta^*$. The simultaneous perturbation stochastic approximation (SPSA) 
algorithm is a \emph{noisy gradient} algorithm which works as illustrated in Figure \ref{noisegrad}.
It requires only $2$ observations per iteration. 
We adapt it to tune the parameters of Hadoop. By adaptively tuning the Hadoop parameters,
we intend to optimize the Hadoop job execution time, which is the performance metric (i.e., $f(\theta)$) in 
our experiments. Note that we can also have other performance metrics - like number of records spilled to disk,
Memory and heap usage or number of failed jobs. 
Next, we provide a detailed description of SPSA.

\subsection{Simultaneous Perturbation \\Stochastic Approximation (SPSA)}
\label{spsa}
We use the following notation:
\begin{enumerate}
\item $\theta\in X \subset \R^n$ denotes the tunable parameter. Here $n$ is the 
dimension of the parameter space. Also, $X$ is assumed to be a 
compact and convex subset of $\R^n$.
\item Let $x\in \R^n$ be any vector then $x(i)$ denotes its $i^{th}$ 
co-ordinate, i.e., $x=(x(1),\ldots,x(n))$.
\item $f(\theta)$ denotes the performance of the system for parameter $\theta$. 
Let $f$ be a smooth and differentiable function of $\theta$.
\item $\nabla f(\theta)=(\frac{\partial 
f}{\partial\theta(1)},\ldots,\frac{\partial f}{\partial\theta(n)})$ is the 
gradient of the function, and $\frac{\partial f}{\partial \theta(i)}$ is the 
partial derivative of $f$ with respect to $\theta(i)$.
\item $e_i\in \R^n$ is the standard $n$-dimensional unit vector with $1$ in the $i^{th}$ 
co-ordinate and $0$ elsewhere.
\end{enumerate}
Formally the gradient is given by
\begin{align}\label{partial}
\frac{\partial f}{\partial\theta(i)}=\lim_{h\rightarrow 0}\frac{f(\theta+h 
e_i)-f(\theta)}{h}.
\end{align}
In \eqref{partial}, the $i^{th}$ partial derivative is obtained by perturbing 
the $i^{th}$ co-ordinate of the parameter alone and keeping rest of the 
co-ordinates the same. Thus, the number of operations required to compute the 
gradient once via perturbations is of $n+1$. This can be a shortcoming in cases when 
it is costly (i.e., computationally expensive) to obtain measurements of $f$ and 
the number of parameters is large.

The SPSA algorithm \cite{spall} computes the gradient of a function with only 
$2$ or fewer perturbations. Thus the SPSA algorithm is extremely useful in cases 
when the dimensionality is high and the observations are costly. The idea behind 
the SPSA algorithm is to perturb not just one co-ordinate at a time but all the 
co-ordinates together simultaneously in a random fashion. However, one has to 
carefully choose these random perturbations so as to be able to compute the 
gradient. Formally, a random perturbation $\Delta\in \R^n$ should satisfy the 
following assumption.
\begin{assumption}\label{indep}
For any $i\ne j, i=1,\ldots,n, j=1,\ldots,n$, the random variables $\Delta(i)$ 
and $\Delta(j)$ are zero-mean,independent, and the random variable $Z_{ij}$ given by 
$Z_{ij}=\frac{\Delta(i)}{\Delta(j)}$ is such that $\mathbf{E}[Z_{ij}]=0$ and it 
has finite second moment.
\end{assumption}

We now provide an example of random perturbations that satisfies the 
Assumption~\ref{indep}.
\begin{example}
$\Delta\in \R^n$ is such that, each of its co-ordinates $\Delta(i)$s are 
independent Bernoulli random variables taking values $-1$ or $+1$ with equal 
probability, i.e., $Pr\{\Delta(i)=1\}=Pr\{\Delta(i)=-1\}=\frac{1}{2}$ for all 
$i=1,\ldots,n$.
\end{example}
\subsection{Noisy Gradient Recovery from Random Perturbations}
Let $\hat{\nabla} f_\theta$ denote the gradient estimate, and let $\Delta\in 
\R^n$ be any perturbation vector satisfying Assumption~\ref{indep}. Then for any 
small positive constant $\delta >0$, the one-sided SPSA algorithm \cite{chen96,chen97} obtains an estimate of 
the gradient according to equation \eqref{gradest} given below.
\begin{align}\label{gradest}
\hat{\nabla} f_\theta(i)=\frac{f(\theta+\delta 
\Delta)-f(\theta)}{\delta\Delta(i)}.
\end{align}
We now look at the \emph{expected} value of $\hat{\nabla} f_{\theta}(i)$, i.e., 
\begin{align}\label{gradestexp}
\mathbf{E}[\hat{\nabla} f_\theta(i) | \theta]&= \mathbf{E}\left[\frac{f(\theta)+\delta\Delta^\top\nabla f(\theta)+o(\delta^2)-f(\theta)}{\delta \Delta(i)} \;| \theta \right]\nn\\
				    &=\frac{\partial f}{\partial \theta(i)}+\mathbf{E} \left[\sum_{j=1,j\neq i}^n  \frac{\partial f}{\partial \theta(j)}\frac{\Delta(j)}{\Delta(i)} | \theta \right]+o(\delta)\nn\\
				    &=\frac{\partial f}{\partial \theta(i)}+o(\delta).
\end{align}
The third equation follows from the second by noting that 
$\mathbf{E}\left[\frac{\partial f}{\partial \theta(j)}\frac{\Delta(j)}{\Delta(i)} | \theta \right]=0$, 
a fact that follows from the property 
of $\Delta$ in Assumption~\ref{indep}. Thus $\mathbf{E}[\hat{\nabla} 
f_\theta(i)] \ra \nabla f_\theta(i)$ as $\delta\ra 0$.\par
Notice that in order to compute the gradient $\nabla f_\theta$ at the point 
$\theta$ the SPSA algorithm requires only $2$ measurements namely $f(\theta)$ and 
$f(\theta+\delta\Delta)$. 
An extremely useful consequence is that the gradient estimate is not affected by 
the number of dimensions. 

\begin{algorithm}
\caption{Simultaneous Perturbation Stochastic Approximation}
\label{algo}
\begin{algorithmic}[1]
\STATE Let initial parameter setting be $\theta_0\in X\subset\R^n$
\FOR{$n=1,2\ldots,N$}
\STATE Observe the performance of system $f(\theta_n)$.
\STATE Generate a random perturbation vector $\Delta_n \in\R^n$.
\STATE Observe the performance of system $f(\theta_n+\delta\Delta_n)$.
\STATE Compute the gradient estimate $\hat{\nabla} 
f_n(i)=\frac{f(\theta_n+\delta 
\Delta_n)-f(\theta_n)}{\delta\Delta_n(i)}.$\label{gres}
\STATE Update the parameter in the negative gradient direction 
$\theta_{n+1}(i)=\Gamma\big(\theta_n(i)-\alpha_n\frac{
f(\theta_n+\delta\Delta_n)-f(\theta_n)}{\delta\Delta_{n}(i)}\big)$.\label{update}
\ENDFOR
\RETURN $\theta_{N+1}$
\end{algorithmic}
\label{politer}
\end{algorithm}
The complete SPSA algorithm is shown in Algorithm~\ref{algo}, where
$\{\alpha_n\}$ is the step-size schedule and $\Gamma$ 
is a projection operator that keeps the iterates within $X$. We now briefly discuss 
the conditions and nature of the convergence of the SPSA algorithm. 
\begin{figure}[h!]
\centering
\begin{tikzpicture}[domain=0:7.7,scale=0.6,font=\small,axis/.style={very thick, ->, >=stealth'}]
\draw (0,-3) rectangle (10,-2);
\draw[-](-0.5,0)--(0,0);
\draw[-](-0.5,0)--(-0.5,-2.5);
\draw[-](-0.5,-2.5)--(0,-2.5);
\draw[-](0,0.25)--(0,-0.5);

\draw[->](0,0.25)--(3,0.25);
\node [right] at (-0.1,0.5){$\theta_n+\delta\Delta_n$};
\draw (3,0.5)rectangle (5,0);
\node [right] at (3,0.25) {System};
\node [right] at (5.5,0.5) {$f(\theta_n+\delta\Delta_n)$};

\draw[-](5,0.25)--(10.5,0.25);
\draw[-](10.5,0.25)--(10.5,-2.75);
\draw[->](10.5,-2.75)--(10,-2.75);

\draw[->](0,-0.5)--(3,-0.5);
\node [right] at (0,-0.75){$\theta_n$};
\draw (3,-0.75)rectangle (5,-0.25);
\node [right] at (3,-0.5) {System};
\node [right] at (5.5,-0.75) {$f(\theta_n)$};

\draw[-](5,-0.5)--(10.25,-0.5);
\draw[-](10.25,-0.5)--(10.25,-2.25);
\draw[->](10.25,-2.25)--(10,-2.25);

\node [right] at (-0.1,-2.5){$\theta_{n+1}(i)=\Gamma\big(\theta_n(i)-\alpha_n\frac{f(\theta_n+\delta\Delta_n)-f(\theta_n)}{\delta\Delta(i)}\big)$};
\node [rotate=90] at (-0.75,-1) {Parameter $\theta_n$};

\end{tikzpicture}
\label{schspsa}
\caption{SPSA. Gradient is computed by perturbing all the co-ordinates. Notice the \emph{sytem-in-loop} nature of the SPSA based tuning.}
\end{figure}
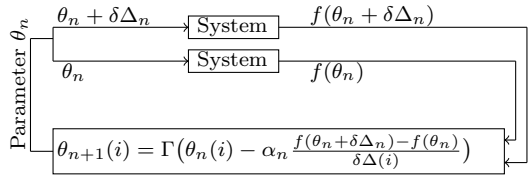 

\subsection{Convergence Analysis}
The SPSA algorithm (Algorithm~\ref{algo}) makes use of a noisy gradient estimate 
(in line~\ref{gres}) and at each iteration takes a step in the negative gradient 
direction so as to minimize the cost function. The noisy gradient update can be 
re-written as 
\begin{align}\label{stocrec}
\theta_{n+1}&= \Gamma \left(\theta_n-\alpha_n\big(\mathbf{E}[\hat{\nabla}f_n | \theta_n] +\hat{\nabla}f_n -\mathbf{E}[\hat{\nabla}f_n | \theta_n ])\right)\\
	    &= \Gamma \left(\theta_n-\alpha_n\big(\nabla f_n + M_{n+1}+\epsilon_n )\right)\nn
\end{align}
where $M_{n+1}=\hat{\nabla}f_n -\mathbf{E}[\hat{\nabla}f_n | \theta_n]$ is an associated martingale 
difference sequence under the sequence of $\sigma$-fields 
$\mathfrak{F}_n = \sigma (\theta_m, m \leq n, \Delta_m, m < n), \; n \geq 1$ 
and $\epsilon_n$ is a small bias due to the $o(\delta)$ term 
in \eqref{gradestexp}.

The iterative update in \eqref{stocrec} is known as a stochastic approximation 
\cite{SA} recursion. As per the theory of stochastic approximation, in order to 
\emph{filter} out the noise, the step-size schedule $\{\alpha_n\}$ needs to 
satisfy the conditions below.
\begin{align}\label{ss}
\sum_{n=0}^\infty \alpha_n =\infty, \sum_{n=0}^\infty \alpha_n^2<\infty.
\end{align}
We now state the convergence result.
\begin{theorem}\label{conv}
As $n\ra \infty$ and $\delta\ra 0$, the iterates in \eqref{stocrec} (i.e., 
line~\ref{update} of Algorithm~\ref{algo}) converge to a set $A=\{\theta|\Gamma (\nabla 
f(\theta) ) =0\}$, where for any continuous 
$J: \R^n \ra \R^n$, $ \hat{\Gamma}(J(x)) = \lim\limits_{\eta \downarrow 0} \left( \frac{\Gamma (x + \eta J(x))- \Gamma (x)}{\eta} \right)$ .
\end{theorem}
\begin{proof}
 The proof is similar to Theorem 3.3.1, pp. 191-196 of \cite{kushner}.
\end{proof}

Theorem~\ref{conv} guarantees the convergence of the iterates to local minima. 
However, in practice local minima corresponding to small valleys are avoided due 
either to the noise inherent to the update or one can periodically inject some 
noise so as to let the algorithm explore further. Also, though the result stated in 
Theorem~\ref{conv} is only asymptotic in nature, in most practical cases 
convergence is observed in a finite number of steps. In the following section, we 
adapt SPSA to the problem of parameter tuning for enhancing the performance of Hadoop.

\section{Applying SPSA to Hadoop Parameter tuning}
\label{setup}
The SPSA algorithm was presented in its general form in Algorithm~\ref{algo}. We 
now discuss the specific details involved in suitably applying the SPSA 
algorithm to the problem of parameter tuning in Hadoop.

\subsection{Mapping the Parameters}
The SPSA algorithm needs each of the parameter components to be real-valued i.e., in 
Algorithm~\ref{algo}, $\theta\in X \subset\R^n$. However, most of the Hadoop 
parameters that are of interest are not $\R^n$-valued. Thus, on the one hand we 
need a set of $\R^n$-valued parameters that the SPSA algorithm can tune and a 
mapping that takes these $\R^n$-valued parameters to the Hadoop parameters. In 
order to make things clear we introduce the following notation:

\begin{enumerate}
\item The Hadoop parameters are denoted by $\thetah$ and the $\R^n$-valued
parameters tuned by SPSA are denoted by $\thetaa$\footnote{Here subscripts $A$ 
and $H$ are abbreviations of the keywords Algorithm and Hadoop respectively}.
\item $S_i$ denotes the set of values that the $i^{th}$ Hadoop parameter can 
assume. $\thetahn(i)$, $\thetahx(i)$ and $\thetahd(i)$ denote the 
\emph{minimum}, \emph{maximum} and \emph{default} values that the $i^{th}$ 
Hadoop parameter can assume.
\item $\thetaa \in X\subset\R^n$ and $\thetah \in S_1\times\ldots \times S_n$. 
\item $\thetah=\mu(\thetaa)$, where $\mu$ is the function that maps $\thetaa\in 
X\subset \R^n$ to $\thetah \in S_1\times\ldots \times S_n$. 
\end{enumerate}

In this paper, we choose $X=[0,1]^n$, and $\mu$ is defined as 
$\mu(\thetaa)(i)=\lfloor(\thetahx(i)-\thetahn(i))\thetaa(i)+\thetahn(i)\rfloor$ 
and 
$\mu(\thetaa)(i)=(\thetahx(i)-\thetahn(i))\thetaa(i)+\thetahn(i)$ for integer 
and real-valued Hadoop parameters respectively.

\subsection{Perturbation Sequences and Step-Sizes}
We chose $\delta\Delta_n \in \R^n$ be independent random variables, such that 
$Pr\{\delta\Delta_n(i)=-\frac{1}{\thetahx(i)-\thetahn(i)}\} 
=Pr\{\delta\Delta_n(i)=+\frac{1}{\thetahx(i)-\thetahn(i)}\}=\frac{1}{2}$.
This perturbation sequence ensures that the Hadoop parameters assuming only 
integer values change by a magnitude of at least $1$ in every perturbation. 
Otherwise, using a perturbation whose magnitude is less than 
$\frac{1}{\thetahx(i)-\thetahn(i)}$ might not cause any change to the 
corresponding Hadoop parameter resulting in an incorrect gradient estimate.

The conditions for the step-sizes in \eqref{ss} are asymptotic in nature and are 
required to hold in order to be able to arrive at the result in 
Theorem~\ref{conv}. However, in practice, a constant step size can be used since 
one reaches closer to the desired value in a finite number of iterations. 
We know apriori that the parameters tuned by the SPSA algorithm belong to the 
interval $[0,1]$ and it is enough to have step-sizes of the order of 
$\min_i(\frac{1}{\thetahx(i)-\thetahn(i)})$ (since any finer step-size used to 
update the SPSA parameter $\thetaa(i)$ will not cause a change in the 
corresponding Hadoop parameter $\thetah(i)$). In our experiments, we chose 
$\alpha_n=0.01, \forall n\geq 0$ and observed convergence in about $20$ 
iterations.

\section{Experimental Evaluation}
\label{exp}
We use Hadoop versions 1.0.3 and 2.7 in our experiments. The SPSA
algorithm described in Sections \ref{autotune},\ref{setup} is implemented
as a process which executes in the Resource Manager (and/or NameNode).
First we justify the selection of parameters to be tuned in our experiments, following which
we give details about the implementation.

\subsection{Parameter Selection}

As discussed in Section \ref{hadoop}, based on the data flow analysis of Hadoop MapReduce
and the Hadoop manual \cite{url_4}, we identify 11 parameters which are found to critically
affect the operation of HDFS and the Map/Reduce operations. 
The list of important parameters that emerged by analyzing the MapReduce implementation are 
listed in Table \ref{tab:opt_parameter}.
Numerous parameters of Hadoop deal with book keeping related tasks, whereas some other parameters
are connected with the performance of underlying operating system tasks.
For e.g., \textit{mapred.child.java.opts} is a parameter related to the Java virtual machine (JVM) of 
Hadoop. We avoid tuning such parameters, which are best left for low-level OS optimization.
Instead we tune parameters which are directly Hadoop dependent, for e.g., number of reducers, I/O utilization
parameter etc. 
However, even with $11$ parameters, the search space is still large and complex. 
To see this, if each parameter can assume say $10$ different values then the search space 
contains $10^{11}$ possible parameter settings. 
Some parameters have been left out because 
either they are  backward incompatible or they incur additional overhead in implementation. 

\subsection{Cluster Setup}
Our Hadoop cluster consists of $25$ nodes. Each node has a 8 core Intel Xeon E3, 2.50 GHz 
processor, $3.5$ TB HDD, $16$ GB memory, 1 MB L2 Cache, 8MB L3 Cache. 
One node works as the 
NameNode and the rest of the nodes are used as DataNodes.
For optimization and evaluation purpose we set the number of map slots to $3$ and 
reduce slots to $2$ per node. Hence, in a single wave of Map jobs processing, the cluster can process 
$24 \times 3 = 72$ map tasks  and $24 \times 2 = 48$ reduce tasks (for more details see \cite{hadoop}). 
HDFS block replication was set to $2$. We use a dedicated Hadoop cluster in our experiments, which
is not shared with any other application.

\subsection{Benchmark Applications}
In order to evaluate the performance of tuning algorithm, we adapt representative MapReduce applications.
The applications we use are listed in Table \ref{tab:opt_parameter}.
Terasort application takes as input a text data file and sorts it. It has three components - 
\emph{TeraGen} - which generates the input data for sorting algorithm, \emph{TeraSort} - algorithm that 
implements sorting and \emph{TeraValidate} - validates the sorted output data. 
The Grep application searches for a particular pattern in a given input file.
The Word Cooccurrence application counts the number of occurrences of a particular word
in an input file (can be any text format).
Bigram application counts all unique sets of two consecutive words in a set of documents, 
while the Inverted index application generates word to document indexing from a list of documents. 
Word Co-occurrence is a popular Natural Language Processing program which
computes the word co-occurrence matrix of a large text collection.
As can be inferred, Grep and Bigram applications are CPU intensive, while
the Inverted Index and TeraSort applications are both CPU and memory intensive.
These benchmark applications can be further categorized as map-intensive, reduce-intensive etc.

\subsection{SPSA Iterations}
SPSA is an iterative algorithm and it runs a Hadoop job with different configurations. 
We refer to these iterations as the \emph{optimization} or the \emph{learning} phase.
The algorithm eventually converges to an optimal value of the configuration parameters. 
The performance metric (the job execution time) corresponding to the converged parameter vector is optimal 
for the corresponding application.
During our evaluations we have seen that SPSA 
converges within 20 - 30 iterations and within each iteration 
it makes two observations, i.e. it executes Hadoop job twice taking the total count of 
Hadoop runs during the optimization phase to 40 - 60. It is of utmost importance 
that the optimization phase is fast otherwise it can overshadow the benefits which 
it provides.

In order to ensure that the optimization phase is fast, we execute the Hadoop jobs 
on a subset of the workload. This is done, since SPSA takes considerable time when
executed on large workloads.
Deciding the size of this ``partial workload'' is very important
as the run time on a small work load will be eclipsed by the job setup and cleanup 
time.  We then consider the technique involved in processing done by Hadoop system 
to find a suitable workload size. Hadoop splits the input data based on the 
block size of HDFS. It then spawns a map for each of the splits and processes each of 
the maps in parallel. The number of the map tasks that can run at a given time is 
upper bounded by the total map slots available in the cluster. Using this fact, the 
size of the partial data set which we use is equal to twice the number of 
map slots in the cluster multiplied by the data block size. Hadoop finishes the 
map task in two waves of the maps jobs which allows the SPSA to capture the 
statistics with a single wave and the cross relations between two successive waves.

Our claim is that the value of configuration parameters which optimize these two 
waves of Hadoop job also optimize all the subsequent waves as those are repetitions of 
similar map jobs. However, the number of reducers to run is completely optimized by SPSA, albeit
for a partial workload size. For the larger (actual) workload, the number of 
reducers decided is based on the ratio of partial work load size to the actual size of workload.
An advantage with SPSA iterations is that these can be halted at any parameter configuration
(for e.g., need for executing a production job on the cluster) and later resumed at the same parameter
configuration where the iterations were halted.

\subsection{Optimization Settings}
For evaluating performance of SPSA on different benchmarks, two waves of map tasks during job execution were ensured.
Further, we selected workloads such that the execution time with default configuration is at least $10$ minutes.
This was done to avoid the scenario where the job setup and cleanup time overshadows 
the actual running time and there is practically nothing for SPSA to optimize.

In the cases of Bigram and Inverted Index benchmark executions, we observed that even with small 
amount of data, the job completion time is high (since they are reduce-intensive operations).
So, due to this reason, we used small sized input data files. Using small sized input data files
resulted in the absence of two waves of map tasks. However, since in these applications,
reduce operations take precedence, the absence of two waves of map tasks did not create
much of a hurdle.

We optimize Terasort using a partial data set of size \textit{30GB}, Grep on \textit{22GB},
Word Co-occurrence on \textit{85GB}, Inverted Index on \textit{1GB} and Bigram count on \textit{200MB}
of data set. In optimization (or learning) phases, for each benchmark, we use the default configuration as the initial point
for the optimization. Table \ref{tab:opt_parameter} indicates
the default values and the values provided by SPSA for the parameters we tune.
For greater sizes of data, we used Wikipedia dataset\cite{puma}($\approx100$GB) for Word co-occurrence, Grep and
Bigram benchmarks

The SPSA algorithm is terminated when either the change in gradient estimate is negligible or 
the maximum number of iterations have been reached. An important point to note is that Hadoop parameters can take
values only in a fixed range. We take care of this by projecting the tuned parameter values 
into the range set (component-wise). A penalty-function can also be used instead. 
If noise level in the function evaluation is high then, 
 it is useful to average several SP gradient estimates (each with 
independent values of $\Delta_{k}$) at a given  $\theta_{k}$.  Theoretical 
justification for net improvements to efficiency by such gradient averaging is 
given in ~\cite{Spall92multivariatestochastic}. We can also use a one evaluation variant of SPSA,
which can reduce the per iteration cost of SPSA. 
But it has been shown  that standard two function  measurement form, 
which we use in our work is more efficient (in terms of total number of loss 
function measurements) to obtain a given level of accuracy in the $\theta$ 
iterate. 

\subsection{Comparison with Related Work}
We compare our method with prior works in the literature on Hadoop performance tuning.
Specifically, we look at Starfish\cite{starfish} as well as Profiling and Performance Analysis-based System 
(PPABS) \cite{ppabs} frameworks.
We briefly describe these methods in Section \ref{relatedwork}. Starfish is designed for Hadoop version 1
only, whereas PPABS works with the recent versions also. Hence, in our experiments we use both versions
of Hadoop. To run Starfish, we use the executable hosted by the authors of \cite{starfish} to profile
the jobs run on partial workloads. Then execution time of new jobs is obtained by running the jobs 
using parameters provided by Starfish. For testing PPABS, we collect datasets as described in \cite{ppabs},
cluster them and find optimized parameters (using simulated annealing) for each cluster.
Each new job is then assigned to one cluster and executed with the parameters optimized for that cluster.

\subsection{Discussion of Results}
\begin{table*}[t!]
\small
\centering
{\renewcommand{\arraystretch}{1.5}%
    \begin{tabular}{|l|c|c|c|c|c|c|c|c|c|c|c|c}
    \hline
\multirow{2}{*}{Parameter Name} & 
\multirow{2}{*}{Default} &
\multicolumn{2}{c|}{Terasort} &
\multicolumn{2}{c|}{Grep} &
\multicolumn{2}{c|}{Bigram}  &
\multicolumn{2}{c|}{Inverted Index} &
\multicolumn{2}{c|}{Word Co-occurrence}\\
\cline{3-12}

& &  v1.0.3  &  v2.6.3& v1.0.3  &  v2.6.3&  v1.0.3  &  v2.6.3& v1.0.3  &  v2.6.3 & v1.0.3  &  v2.6.3\\
    \hline
io.sort.mb                      & 100	& 149 		& 524	& 50		&	&\textbf{751}	& 779	&1609	& 202	&221 	& 912   \\
io.sort.spill.percent           & 0.08 	& 0.14		& 0.89	& \textbf{0.83}	&	& 0.53		& 0.53	&0.2	& 0.68	&0.75 	& 0.47  \\
io.sort.factor                  & 10 	& \textbf{475}	& 115	& 5		&	& 5		& 178	&50	& 85	&40 	& 5	\\
shuffle.input.buffer.percent 	& 0.7   & 0.86		& 0.87	& 0.67		&	& 0.43		& 0.43	&0.83	& 0.58	&0.65 	& 0.37  \\
shuffle.merge.percent       	& 0.66 	& 0.14		& 0.83	& 0.63		&	& 0.89		& 0.39	&0.83	& 0.54	&0.71 	& 0.33	\\
inmem.merge.threshold           & 1000  & \textbf{9513}	& 318	& \textbf{681}	&	& 4201		& 200	&1095	& 948	&1466 	& 200	\\
reduce.input.buffer.percent 	& 0.0  	& 0.14		& 0.19	& 0.13		&	& 0.31		& 0.0	&0.17	& 0.07	&0.12 	& 0.0 	\\
mapred.reduce.tasks             & 1 	& \textbf{95}	& 22	& 1		&	&\textbf{33}	& 35	&76	& 16	&14 	& 41	\\
io.sort.record.percent 		& 0.05 	& 0.14		& -	& 0.1		&	& 0.31		& -	 &0.17 	& -	&0.2 	& -	\\
mapred.compress.map.output	& false & true		& -	& false		&	& false		& -	&true 	& -	&false 	& -	\\
mapred.output.compress		& false & false		& -	& false		&	& false		& -	&false	& -	&false 	& -	\\
reduce.slowstart.completedmaps	& 0.05	& - 		& 0.23  & -		&	& -		& 0.05	& -	& 0.18	& -	& 0.4	\\
mapreduce.job.jvm.numtasks	& 1	& -		& 2	& -		&	& -		& 18	& -	& 5	& -	& 21	\\
mapreduce.job.maps		& 2	& -		& 23	& -		&	& -		& 35	& -	& 17	& -	& 2	\\
\hline
 \end{tabular}}
  \caption{ Default value of parameters and their values tuned by SPSA (the last three parameters are defined for Hadoop v2)}
\label{tab:opt_parameter}
\end{table*}

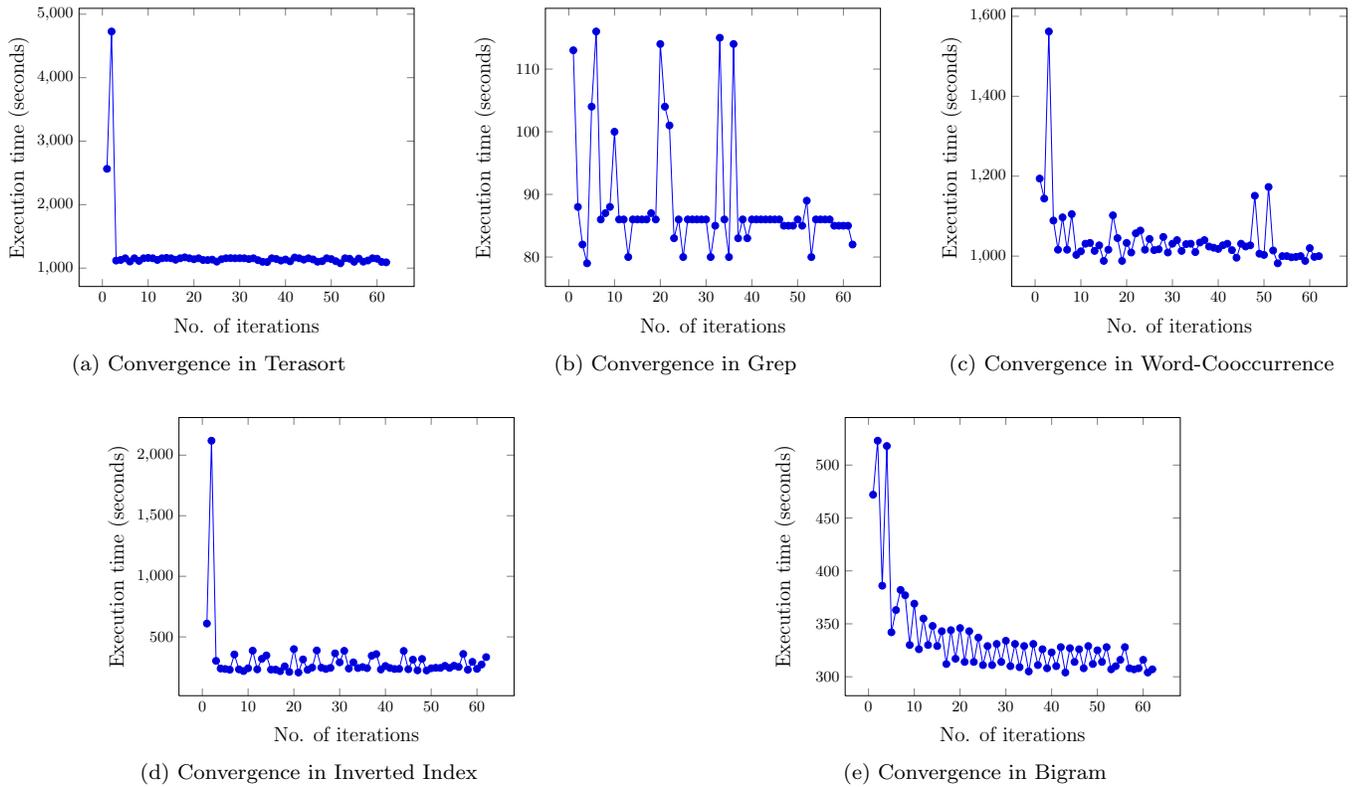
\begin{figure*}[!htb]
\begin{subfigure}[b]{.35\textwidth}
\centering
\begin{tikzpicture}[scale=0.65]
  \begin{axis}[ ylabel= {\large Execution time (seconds)},xlabel= {\large No. of iterations}]
  \addplot table[x expr=\coordindex+1,y index=0] {terasort.csv};  
  \end{axis}
\end{tikzpicture}
\caption{Convergence in Terasort}
\end{subfigure}%
\begin{subfigure}[b]{.35\textwidth}
\centering
\begin{tikzpicture}[scale=0.65]
  \begin{axis}[ ylabel= {\large Execution time (seconds)},xlabel= {\large No. of iterations}]
  \addplot table[x expr=\coordindex+1,y index=0] {grep.csv};
  \end{axis}
\end{tikzpicture}
\caption{Convergence in Grep}
\end{subfigure}%
\begin{subfigure}[b]{.35\textwidth}
\centering
\begin{tikzpicture}[scale=0.65]
  \begin{axis}[ ylabel= {\large Execution time (seconds)},xlabel= {\large No. of iterations}]
  \addplot table[x expr=\coordindex+1,y index=0] {wordco.csv};
  \end{axis}
\end{tikzpicture}
\caption{Convergence in Word-Cooccurrence}
\end{subfigure}%

\vspace{0.5cm}

\begin{subfigure}[b]{.5\textwidth}
\centering
\begin{tikzpicture}[scale=0.65]
  \begin{axis}[ ylabel={\large Execution time (seconds)},xlabel= {\large No. of iterations}]
  \addplot table[x expr=\coordindex+1,y index=0] {ii.csv};
  \end{axis}
\end{tikzpicture}
\caption{Convergence in Inverted Index}
\end{subfigure}%
\begin{subfigure}[b]{.5\textwidth}
\centering
\begin{tikzpicture}[scale=0.65]
  \begin{axis}[ ylabel= {\large Execution time (seconds)},xlabel= {\large No. of iterations}]
  \addplot table[x expr=\coordindex+1,y index=0] {bigram.csv};
  \end{axis}
\end{tikzpicture}
\caption{Convergence in Bigram}
\end{subfigure}%
  \caption{Performance of SPSA for different benchmarks on Hadoop v1}
  \label{fig:spsaConv}
\end{figure*}
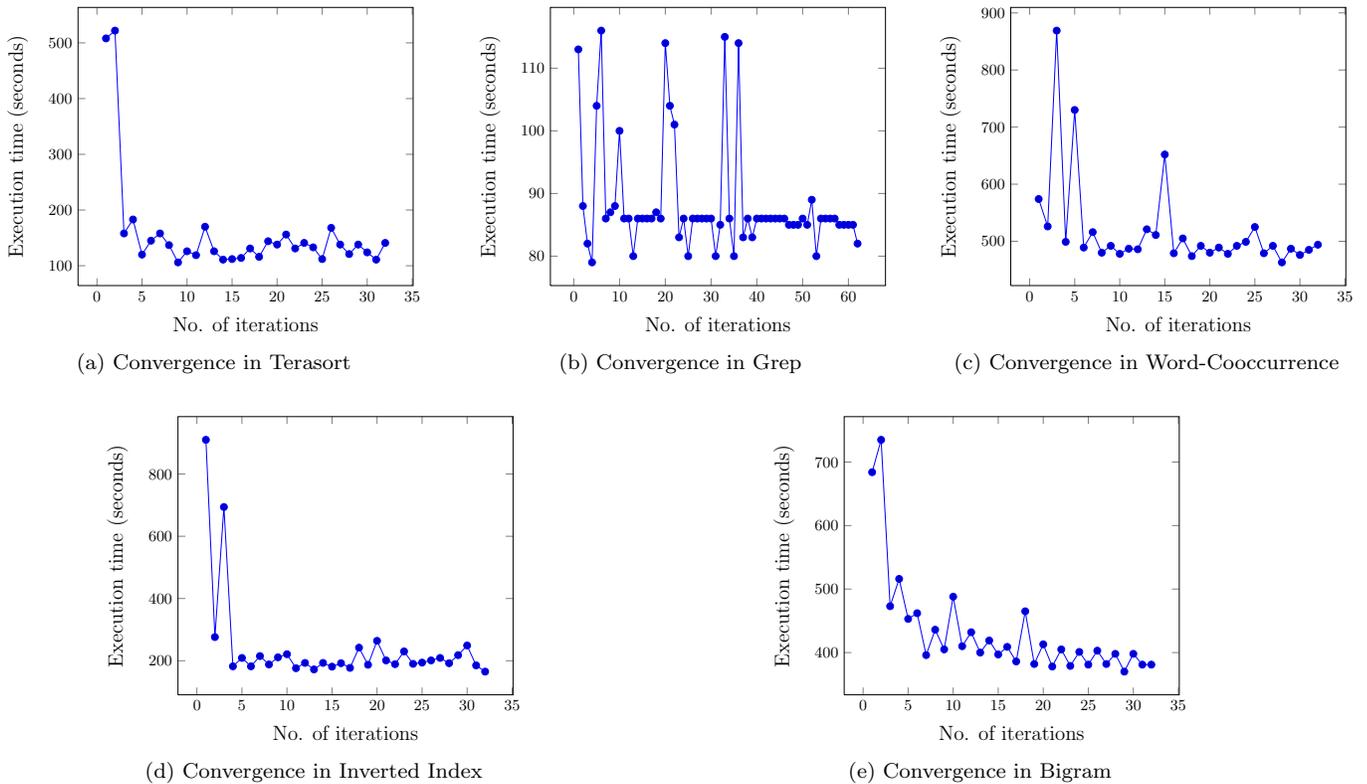
\begin{figure*}[!htb]
\begin{subfigure}[b]{.35\textwidth}
\centering
\begin{tikzpicture}[scale=0.65]
  \begin{axis}[ ylabel= {\large Execution time (seconds)},xlabel= {\large No. of iterations}]
    \addplot table[x expr=\coordindex+1,y index=0] {terasort-v2.csv};
  \end{axis}
\end{tikzpicture}
\caption{Convergence in Terasort}
\end{subfigure}%
\begin{subfigure}[b]{.35\textwidth}
\centering
\begin{tikzpicture}[scale=0.650]
  \begin{axis}[ ylabel= {\large Execution time (seconds)},xlabel= {\large No. of iterations}]
  \addplot table[x expr=\coordindex+1,y index=0] {grep.csv};
  \end{axis}
\end{tikzpicture}
\caption{Convergence in Grep}
\end{subfigure}%
\begin{subfigure}[b]{.35\textwidth}
\centering
\begin{tikzpicture}[scale=0.650]
  \begin{axis}[ ylabel= {\large Execution time (seconds)},xlabel= {\large No. of iterations}]
  \addplot table[x expr=\coordindex+1,y index=0] {wordco-v2.csv};
  \end{axis}
\end{tikzpicture}
\caption{Convergence in Word-Cooccurrence}
\end{subfigure}%

\vspace{0.5cm}

\begin{subfigure}[b]{.5\textwidth}
\centering
\begin{tikzpicture}[scale=0.650]
  \begin{axis}[ ylabel= {\large Execution time (seconds)},xlabel= {\large No. of iterations}]
  \addplot table[x expr=\coordindex+1,y index=0] {ii-v2.csv};
  \end{axis}
\end{tikzpicture}
\caption{Convergence in Inverted Index}
\end{subfigure}%
\begin{subfigure}[b]{.5\textwidth}
\centering
\begin{tikzpicture}[scale=0.650]
  \begin{axis}[ ylabel= {\large Execution time (seconds)},xlabel= {\large No. of iterations}]
  \addplot table[x expr=\coordindex+1,y index=0] {bigram-v2.csv};
  \end{axis}
\end{tikzpicture}
\caption{Convergence in Bigram}
\end{subfigure}%
  \caption{Performance of SPSA for different benchmarks on Hadoop v2}
  \label{fig:spsaConv2}
\end{figure*}
\begin{figure*}

\begin{subfigure}[b]{.33\textwidth}
\centering
\begin{tikzpicture}[scale=0.55]
\begin{axis}[
    ybar,
    enlargelimits=0.15,
    legend style={at={(0.5,-0.15)},
      anchor=north,legend columns=-1},
    ylabel={{\Large Execution Time (seconds)}},
    symbolic x coords={30GB,50GB,100GB},
    xtick=data,
    ]
\addplot[black,pattern=north east lines] coordinates {(30GB,3140) (50GB,3556) (100GB,7223)};

\addplot[black] coordinates {(30GB,2821.5) (50GB,3956) (100GB,7200)};

\addplot[black,pattern=horizontal lines] coordinates {(30GB,115.5) (50GB,426) (100GB,860)};

\legend{Default,Starfish,SPSA}
\end{axis}
\end{tikzpicture}




\caption{Terasort}
\end{subfigure}
\begin{subfigure}[b]{0.33\textwidth}
\centering
\begin{tikzpicture}[scale=0.55]
\begin{axis}[
    ybar,
    enlargelimits=0.15,
    legend style={at={(0.5,-0.15)},
      anchor=north,legend columns=-1},
    ylabel={{\Large Execution Time (in seconds)}},
    symbolic x coords={200MB,2GB,10GB},
    xtick=data,
    ]
\addplot[black,pattern=north east lines] coordinates {(200MB,518) (2GB,2460) (10GB,14220)};

\addplot[black] coordinates {(200MB,306.5) (2GB,865) (10GB,2582)};

\addplot[black,pattern=horizontal lines] coordinates {(200MB,240) (2GB,746) (10GB,1744)};

\legend{Default,Starfish,SPSA}
\end{axis}
\end{tikzpicture}



\caption{Bigram}
\end{subfigure}%
\begin{subfigure}[b]{.33\textwidth}
\centering
\begin{tikzpicture}[scale=0.55]
\begin{axis}[
    ybar,
    enlargelimits=0.15,
    legend style={at={(0.5,-0.15)},
      anchor=north,legend columns=-1},
    ylabel={{\Large Execution Time}},
    symbolic x coords={1GB,4GB,8GB},
    xtick=data,
    ]
\addplot[black,pattern=north east lines] coordinates {(1GB,536) (4GB,1712) (8GB,11700)};

\addplot[black] coordinates {(1GB,111.5) (4GB,868) (8GB,1765)};

\addplot[black,pattern=horizontal lines] coordinates {(1GB,100) (4GB,475) (8GB,960)};
\legend{Default,Starfish,SPSA}
\end{axis}
\end{tikzpicture}



\caption{Inverted Index}
\end{subfigure}%

\vspace{0.5cm}
\begin{subfigure}[b]{0.5\textwidth}
\centering
\begin{tikzpicture}[scale=0.55]
\begin{axis}[
    ybar,
    enlargelimits=0.15,
    legend style={at={(0.5,-0.15)},
      anchor=north,legend columns=-1},
    ylabel={{\Large Execution Time}},
    symbolic x coords={50GB,65GB,100GB},
    xtick=data,
    ]
\addplot[black,pattern=north east lines] coordinates { (50GB,124) (65GB,260) (100GB,207)};

\addplot[black] coordinates {(50GB,137)(65GB,200) (100GB,250)};

\addplot[black,pattern=horizontal lines] coordinates {(50GB,110) (65GB,162) (100GB,193)};
\legend{Default,Starfish,SPSA}
\end{axis}
\end{tikzpicture}

\caption{Grep}
\end{subfigure}%
\begin{subfigure}[b]{.5\textwidth}
\centering
\begin{tikzpicture}[scale=0.55]
\begin{axis}[
    ybar,
    enlargelimits=0.15,
    legend style={at={(0.0,-0.15)},
      anchor=north,legend columns=-1},
    ylabel={{\Large Execution Time}},
    symbolic x coords={1GB,2GB,4GB},
    xtick=data,
    ]
\addplot[black,pattern=north east lines] coordinates {(1GB,1249) (2GB,7758) (4GB,15660) };

\addplot[black] coordinates {(1GB,1002) (2GB,3263) (4GB,8023)};

\addplot[black,pattern=horizontal lines] coordinates {(1GB,985) (2GB,3002) (4GB,3732) };
\legend{Default,Starfish,SPSA}
\end{axis}
\end{tikzpicture}


\caption{Word Co-occurrence}
\end{subfigure}%

  \caption{Performance comparison of SPSA, Starfish and Default settings for benchmark applications (MapReduce v1)}
  \label{fig:perf}
\end{figure*}
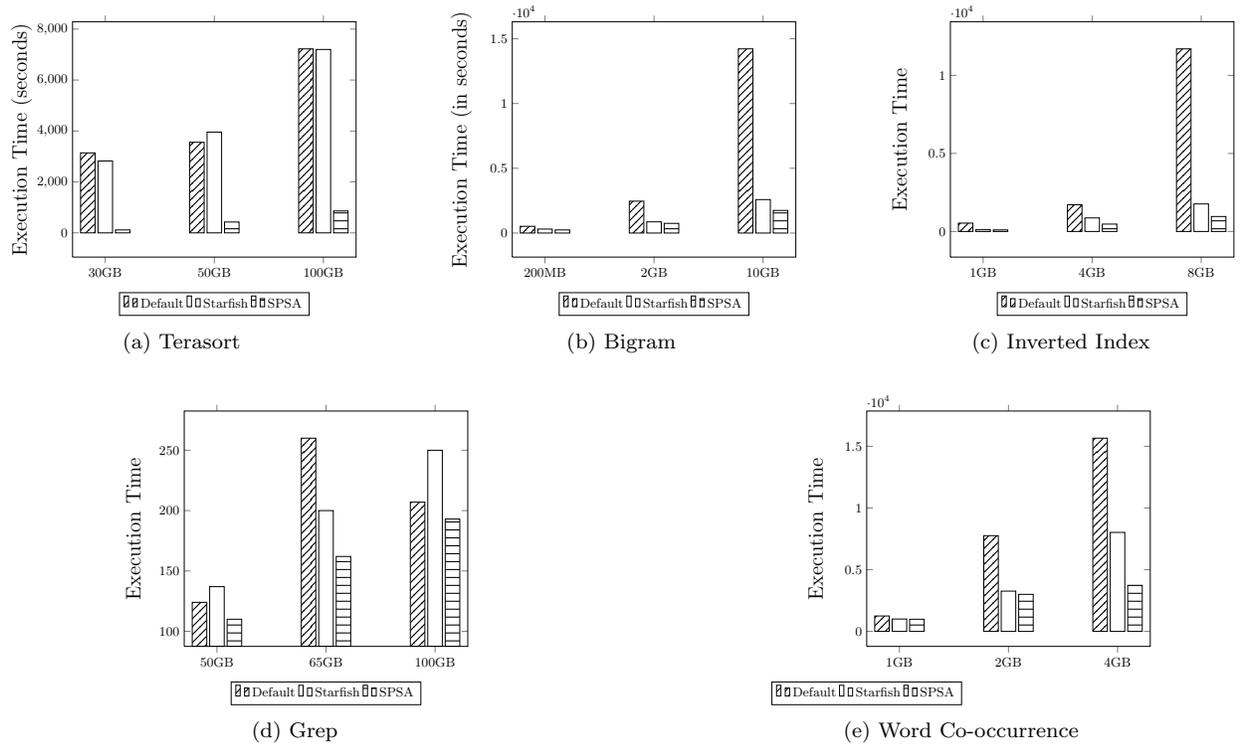
\begin{figure*}
\centering
%
%
%
%
%

\begin{tikzpicture}[scale=0.55]
\begin{axis}[
    ybar,x=5cm,
  legend style={font=\Large,at={(0.5,-0.25)},
      anchor=north,legend columns=-1},
    ylabel={{\Large Execution Time (seconds)}},
    symbolic x coords={Terasort ($50$GB),Inverted Index ($4$GB), Bigram ($2$GB),Word Co-occurrence ($2$GB) },
    xtick=data,tick label style={font=\Large},
    x label style={at={(axis description cs:0.5,-0.1)},anchor=north},
    y label style={at={(axis description cs:-0.05,.5)},anchor=south},
    ]
\addplot[black,pattern=north east lines] coordinates {(Terasort ($50$GB),3562) (Inverted Index ($4$GB),2024) (Bigram ($2$GB),2666) (Word Co-occurrence ($2$GB),8333)};

\addplot[black] coordinates {(Terasort ($50$GB),496) (Inverted Index ($4$GB),442) (Bigram ($2$GB),631) (Word Co-occurrence ($2$GB),2659) };

\addplot[black,pattern=horizontal lines] coordinates {(Terasort ($50$GB),1578) (Inverted Index ($4$GB),683) (Bigram ($2$GB),1127) (Word Co-occurrence ($2$GB),3427)};

\legend{Default,SPSA,PPABS}
\end{axis}
\end{tikzpicture}
  \caption{Performance comparison of Default settings, SPSA and PPABS for benchmark applications (Hadoop v2)}
  \label{fig:perf2}
\end{figure*}
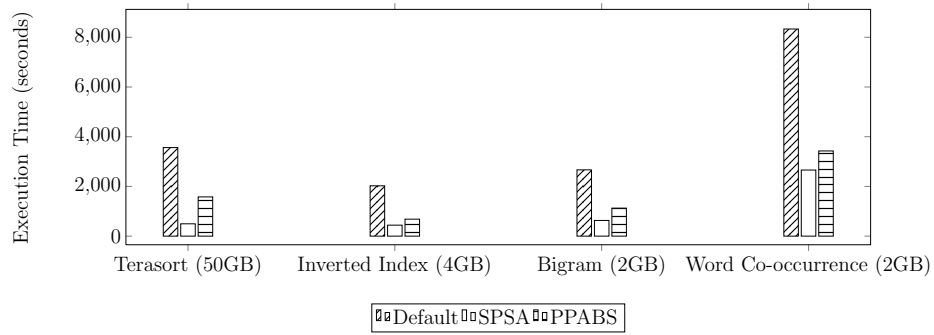

Our method starts optimizing with the default configuration,
hence the first entry in Fig. \ref{fig:spsaConv} show the execution time
of Hadoop jobs for the default parameter setting. 
It is important to note that the jumps in the plots are due to the
noisy nature of the gradient estimate and they eventually
die down after sufficiently large number of iterations.
As can be observed from Fig. \ref{fig:spsaConv}, SPSA reduces the 
execution time of Terasort benchmark by $60-63\%$ when compared to default settings
and by $40-60\%$ when compared to Starfish optimizer.
For Inverted Index benchmark the reduction is $~80\%$ when compared to default settings.
In the case of word co-occurrence, the observed reduction is $22\%$ when compared to default settings
and $2\%$ when compared to Starfish.

SPSA finds the optimal configuration while keeping the relation among the parameters in mind. 
For Terasort, a small value $(0.14)$ of \emph{io.sort.spill.percent} will generate a lot of spilled 
files of small size. Because of this, the value of \emph{io.sort.factor} has been increased to $475$ 
from the default value of $10$. This ensures the combination of number of spilled files 
to generate the partitioned and sorted files.
A small value of \emph{shuffle.input.buffer.percent} $(0.14)$ and a large value of 
\emph{inmem.merge.threshold} $(9513)$ may be confusing as both of them act as a threshold 
beyond which in-memory merge of files (output by map) is triggered.
However, map outputs a total bytes of $100$ GB and a total of $2,000,000,000$ files are spilled 
to disk which, effectively make each spilled file of size $50$ bytes. 
Thus filling $0.14\%$ of the memory allocated to Reduce makes $50$ bytes files of which there will be
$9513$. Default value of number of reducers (i.e., $1$) generally does not work in practical situations. 
However, increasing it to a very high number also creates an issue as it results in more network and disk overhead.
As can be observed in Table \ref{tab:opt_parameter},
mapred.compress.map.output is set to true for Terasort benchmark.
This is because, the output data of Map phase has same size as the input data (which might be huge).
Thus, in such scenarios, it is beneficial if the Map output is compressed.
Grep benchmark, on the other hand produces very little map output, and even smaller sized data to be shuffled. 
Hence \emph{io.sort.mb} value is reduced to $50$ from default $100$ (see Table \ref{tab:opt_parameter})
and number of reducers is set to $1$. Further, value of 
\emph{inmem.merge.threshold} has been reduced to $681$ from $1000$ as there is not much data to work on.

Bigram and Inverted Index are computationally expensive operations.
Hence \emph{io.sort.mb} is increased to $751$ and $617$ respectively. 
Both of these applications also generate a reasonable size of data during the map phase,
which implies a lot of spilled files are generated. 
Thus, \emph{inmem.merge.threshold} has been increased to $4201$ and $3542$ respectively.


\subsection{Advantages of SPSA}
\begin{table*}[t!]
 \small
 \centering
 \begin{tabular}{|c|c|c|c|c|c|}
 \hline
 Methods & Mathematical Model & Dimension & Parameter Dependency & Performance in Real system & Profiling Overhead\\
 \hline
 Starfish & \color{red}{\ding{55}} & \color{red}{\ding{55}} & \color{red}{\ding{55}} & \color{red}{\ding{55}} & \color{red}{\ding{51}}\\ \hline
 PPABS & \color{blue}{\ding{55}} & \color{blue}{\ding{55}} & \color{blue}{\ding{55}} & \color{blue}{\ding{55}} & \color{blue}{\ding{51}}\\ \hline
 SPSA & \ding{51} &\ding{51} & \ding{51} & \ding{51} & \ding{55}\\
  \hline
 \end{tabular}
  \caption{Starfish, PPABS and SPSA : Comparison of Approaches}
 \label{tab:rel}
\end{table*}

The above discussion indicates that SPSA performs well in optimizing Hadoop
parameters. We highlight other advantages (also see Table \ref{tab:rel}) of using 
our proposed method:
\begin{enumerate}
 \item Most of the profiling-based methods (Starfish, MROnline etc), 
 use the internal Hadoop(source code) structure to place ``markers'' for precisely profiling a job. 
 Starfish observes the time spent in each function by using btrace. 
 Small change in the source code make this unusable 
 (clearly Starfish only support Hadoop versions $< 1.0.3$).
 SPSA does not rely on the internal structure of hadoop and only observes the final execution time of the 
 job which can be accessed easily.
 
 \item \textbf{Independent of Hadoop version:} As mentioned previously, profiling-based methods
 are highly dependent on the MapReduce version and any changes in the source code of Hadoop will require
 a version upgrade of these methods. In contrast, our SPSA-based method does not rely on any specific 
 Hadoop version.
 
 \item \textbf{Pause and resume:} SPSA optimizes the parameters iteratively. It starts at a given point (default setting in our case) 
 and then progressively finds a better configuration (by estimating gradient). 
 Such a process can be paused at any iteration and then resumed using the same parameter configuration, where the iteration was stopped.
 This is unlike the profiling-based methods, which need to profile jobs in one go. 
  
 \item SPSA takes into consideration multiple values of execution time of a job for the same parameter setting
 (randomness in execution time).
 This is not the case in other methods, which profile a job only once. 
 Multiple observations helps SPSA to remove the randomness in the job which arise due to the underlying hardware.
 
 \item Parameters can be easily added and removed from the set of tunable parameters, 
 which make our method suitable for scenarios where the user wants to have control over the parameters to be tuned.

 \item \textbf{Profiling overhead:} Profiling of takes a long time (since job run time is not yet optimized 
 during profiling) which adds an extra overhead for these methods like Starfish, PPABS, MROnline etc. For e.g.,
 in our experiments, Starfish profiling executed for $4$ hours, $38$ minutes ($=16680$ seconds) in the case of Word co-occurrence
 benchmark on Wikipedia data of size $4$ GB. Also, Starfish profiled Terasort on $100$ GB of synthetic data for $>2$ hours.
 In contrast, our method does not incur additional ``profiling'' time.
 The SPSA \emph{optimization} time is justified, since each iteration results in learning a better parameter configuration.
\end{enumerate}


%
\section{Conclusions and Future Work}
\label{concl}
Hadoop framework presents the user with a large set of tunable parameters. 
Though default setting is known for these parameters, it is important to tune 
these parameters in order to obtain better performance. However, manual tuning 
of these parameters is difficult owing to the complex nature of the search space 
and the pronounced effect of cross-parameter interactions. This calls for an 
automatic tuning mechanism. Prior attempts at automatic tuning have adopted a 
mathematical model based approach and have resorted to parameter reduction prior 
to optimization. Since, Hadoop is continuously evolving, the mathematical models 
may fail for later versions and given the level of cross parameter interaction, 
it is a good idea to retain as many parameters as possible.

In this paper, we suggested a tuning method based on the simultaneous 
perturbation stochastic approximation (SPSA) algorithm. The salient features of 
the SPSA based scheme included its ability to use observations from a real system 
and its insensitivity to the number of parameters. Also, the SPSA algorithm took 
the cross-parameter interaction into account by computing the gradient at each 
point.
Using the SPSA scheme, we tuned as many as $11$ parameters and observed an 
improvement in execution time on real system. In particular, our experiments on 
benchmark applications such as \emph{Terasort}, \emph{Grep}, \emph{Bigram}, 
\emph{Word Co-occurrence} and 
\emph{Inverted Index} showed that the parameters obtained using the SPSA algorithm 
yielded a decrease of $45$-$66\%$ in execution times on a realistic 25 node cluster.

Our aim here was to introduce the practitioners to an algorithm which was 
different in flavor, simple to implement and was as effective as the previous 
methods.
In this work we considered only Hadoop parameters, however, the SPSA algorithm 
based tuning can include parameters from other layers such OS, System, Hardware 
etc. This will go a long way in providing a holistic approach to performance 
tuning of Hadoop MapReduce. Further, other simulation optimization algorithms like 
\cite{qn-sf,rdsa} can be applied to the problem of Hadoop parameter tuning.


\bibliographystyle{abbrv}
\FloatBarrier
\bibliography{ref}
\end{document}